\documentclass[english,11pt]{article}

\usepackage[german, english]{babel}
\usepackage[pdftex]{color, graphicx}
\usepackage{dsfont}
\usepackage{amsmath}
\usepackage{amssymb}
\usepackage{bm}
\usepackage{url}
\usepackage{dsfont}
\usepackage[utf8]{inputenc}
\usepackage{enumerate}
\usepackage{ifthen}
\usepackage{paralist}
\usepackage{fancyhdr}
\usepackage{subfigure}
\usepackage{pdfpages}
\usepackage[babel, english=american]{csquotes}
\usepackage{makeidx}
\usepackage{geometry}
\usepackage{nccmath}
\usepackage{booktabs, multirow}
\definecolor{myblue}{rgb}{0,0,0.5}
\usepackage[a4paper, plainpages=false, colorlinks, 
    urlcolor=black,%myblue,       % \href{...}{...} external (URL)
    filecolor=black,%green,     % \href{...} local file
    citecolor=black,%blue,     % \cite{...}
    linkcolor=black,%myblue,       % \ref{...} and \pageref{...}
    bookmarks,
    bookmarksnumbered,
    pdfpagemode=UseOutlines,
    pdfcreator = {LaTeX with hyperref package},
    pdfproducer = {pdfLaTeX}]{hyperref}

\newtheorem{theorem}{Theorem}[section]
\newtheorem{corollary}[theorem]{Corollary}
\newtheorem{proposition}[theorem]{Proposition}
\newtheorem{lemma}[theorem]{Lemma}
\newtheorem{remark}[theorem]{Remark}
\newtheorem{observation}[theorem]{Observation}
\newtheorem{definition}[theorem]{Definition}
\newtheorem{claim}[theorem]{Claim}
\newtheorem{problem}[theorem]{Problem}
\newtheorem{question}[theorem]{Question}
\newtheorem{notation}[theorem]{Notation}
\newtheorem{beispiel}[theorem]{Example}

\newtheorem{conjecture}[theorem]{Conjecture}

\newenvironment{theo}[1][\empty]{\begin{theorem} 
\ifthenelse{\equal{#1}{\empty}} {}  {\normalfont(\itshape #1\normalfont)} \normalfont ~\\}{\end{theorem}}
\newenvironment{cor}[1][\empty]{\begin{corollary} 
\ifthenelse{\equal{#1}{\empty}} {}  {\normalfont(\itshape #1\normalfont)}  \normalfont~\\}{\end{corollary}}

\newenvironment{lem}[1][\empty]{\begin{lemma} 
\ifthenelse{\equal{#1}{\empty}} {}  {\normalfont(\itshape #1\normalfont)} \normalfont ~\\}{\end{lemma}}
\newenvironment{rem}[1][\empty]{\begin{remark} 
\ifthenelse{\equal{#1}{\empty}} {}  {\normalfont(\itshape #1\normalfont)} \normalfont ~\\}{\end{remark}}

\newenvironment{defi}[1][\empty]{\begin{definition} 
\ifthenelse{\equal{#1}{\empty}} {}  {\normalfont(\itshape #1\normalfont)} \normalfont ~\\}{\end{definition}}

\newenvironment{prb}[1][\empty]{\begin{problem} 
\ifthenelse{\equal{#1}{\empty}} {}  {\normalfont(\itshape #1\normalfont)} \normalfont ~\\}{\end{problem}}

\newenvironment{nota}[1][\empty]{\begin{notation} 
\ifthenelse{\equal{#1}{\empty}} {} {\normalfont(\itshape #1\normalfont)} \normalfont ~\\}{\end{notation}}

\newenvironment{proof}[1][\empty]{\ifthenelse{\equal{#1}{\empty}} {\paragraph{Proof.}~\\}  
{\paragraph{Proof of #1 }~\\}} {$\hfill \Box$}

\newcommand {\N}[0] {\mathbb{N}}
\newcommand {\R}[0] {\mathbb{R}}
\newcommand {\Q}[0] {\mathbb{Q}}
\newcommand {\Z}[0] {\mathbb{Z}}
\renewcommand {\S}[0] {\mathbb{S}}
\newcommand {\B}[0] {\mathbb{B}}
\newcommand {\NP}[0] {\mathbb{NP}}
\renewcommand{\P}[0] {\mathbb{P}}

\newcommand{\conv}[0] {\mathrm{conv}}
\newcommand{\aff}[0] {\mathrm{aff}}
\newcommand{\lin}[0] {\mathrm{lin}}
\newcommand{\pos}[0] {\mathrm{pos}}
\newcommand{\ext}[0] {\mathrm{ext}}
\newcommand{\epi}[0] {\mathrm{epi}}

\newcommand{\bd}[0]{\mathrm{bd}}
\newcommand{\cl}[0]{\mathrm{cl}}
\renewcommand{\int}[0]{\mathrm{int}}
\newcommand{\relint}[0]{\mathrm{relint}}
\newcommand{\CA}[0]{\mathcal{A}}
\newcommand{\CC}[0]{\mathcal{C}}

\newcommand{\CH}[0]{\mathcal{H}}

\newcommand{\CV}[0]{\mathcal{V}}

\newcommand{\normmax}[1][\empty]{\ifthenelse{\equal{#1}{\empty}} {\textsc{Normmax}}  
{\textsc{Normmax}$_{#1}$}}
\renewcommand{\vector}[1]{\left(\begin{array}{c} #1 \end{array}\right)}

\begin{document}
\selectlanguage{english}
\title{Computational Aspects of the Hausdorff Distance\\ in Unbounded Dimension}
\author{Stefan König\footnote{Institut für Mathematik, Technische Universität Hamburg-Harburg, \texttt{stefan.koenig@tuhh.de}} }
\date{\today}

\maketitle

\begin{center}
\textbf{Abstract.}
\end{center}
\vspace{-0.9cm}
\paragraph{} We study the computational complexity of determining the Hausdorff distance of two polytopes given in halfspace- or vertex-presentation in arbitrary dimension. Subsequently, a matching problem is investigated where a convex body is allowed to be homothetically transformed in order to minimize its Hausdorff distance to another one. For this problem, we characterize optimal solutions, deduce a Helly-type theorem and give polynomial time (approximation) algorithms for polytopes.

\section{Introduction}
\label{sec:introHausdorff}

\paragraph{} The problem of comparing two geometric objects and evaluating their similarity or dissimilarity arises naturally in many applications such as shape fitting, pattern recognition or computer vision. A suitable means that has widely been applied to evaluate the resemblance of two compact sets is the \emph{Hausdorff distance}, see e.g.~\cite{hausdorffImages, hausdorffFaceRecognition, hausdorffComputerVision} or the surveys \cite{altGuibasShapeFitting, shapeFittingSurvey} and the references therein. Consequently, the question of how to compute the Hausdorff distance of geometric objects has already been studied extensively as e.g.~in \cite{abb-amps-95, alt-hausdorff, altGuibasShapeFitting,  atallah-hausdorff}. In addition, often not only the static evaluation problem but also problems where one object is allowed to undergo a transformation from a certain class in order to match the other one are of particular interest, cf. e.g. \cite{hausdorffMatchingTranslation, abb-amps-95, amentaHausdorff, atkinsonPointPatternMatching, hausdorffEuclidean}.

Starting from finite point sets in the plane, there are two main lines of research in the above-mentioned papers: they either study a broader class of geometric objects in the plane as in \cite{parametricSearch, rigidMotions} or different higher dimensional matching problems for finite point sets as in \cite{patterMatchingDspace, exactPointPatternMatching}. For a detailed overview of known results, we refer to the tables in  \cite[Chapter 3]{wenkDiss}.

In view of results as \cite{freundOrlin, gk-93, MegiddoKcenterHard} about computational problems in unbounded dimension, it is a natural and fundamental question to ask whether the Hausdorff distance of two polytopes can be computed efficiently in unbounded dimension. The present paper gives a detailed answer to this question. 

\paragraph{} Of course, the computational complexity of computing the Hausdorff distance of two polytopes depends on their presentation. For the different combinations of presentations of the input polytopes, we classify the problems as being solvable in polynomial time or being hard. In addition to the classic $\NP$-hardness, we also establish W[1]-hardness in the theory of Fixed Parameter Tractability \cite{df-99, fg-06, n-06}. This allows a refined analysis of the influence of the dimension in the $\NP$-hard cases and, loosely speaking, shows that the respective problems have to be considered intractable already in low dimensions. 

We show in Theorem~\ref{theo:HausdorffVVeasy} that if both polytopes are available in vertex presentation, their Hausdorff distance can be computed in polynomial time in arbitrary dimension for the most common norms.  If, on the other hand, at least one polytope is given in halfspace presentation, Theorem~\ref{theo:HausdorffHHard} shows  that evaluating the Hausdorff distance is in general $\NP$- and W[1]-hard. 

\paragraph{} Similarly to the references mentioned above, we also consider a matching problem, which seeks for a homothetic transformation of a convex body that minimizes the Hausdorff distance to another body. Within a general analysis of this problem for convex bodies, we give an optimality criterion in the spirit of John's Theorem \cite{john} in Theorem~\ref{theo:optCondHausdorffHomo}. From this, we deduce a Helly-type result for convex bodies in Theorem~\ref{theo:HMHomoHelly}. Moreover, we demonstrate in Theorem~\ref{theo:matchingSOCP} that for two vertex presented polytopes also the matching problem can be solved efficiently in arbitrary dimension. Interestingly, even if one halfspace presented polytope is involved and the evaluation problem is W[1]-hard, a factor-$(3\sqrt{d}+1)$-approximate solution of the matching problem in $\R^d$ can be computed in polynomial time, cf. Theorem~\ref{theo:referencePoints}.

\paragraph{} The paper is organized as follows. We start by introducing our notation in Section~\ref{sec:preliminaries} and state basic properties of the Hausdorff distance that will be useful at different places throughout the paper. In Section~\ref{sec:hausdorffEval}, we give a precise problem formulation and study the evaluation problem i.e.~we analyze the complexity of lower bounding the Hausdorff distance of two given polytopes. In Section~\ref{sec:hausdorffMatching} we present the results for the matching problem.

\section{Preliminaries}
\label{sec:preliminaries}
\subsection{Notation}

\paragraph{} Throughout this paper, we are working in $d$-dimensional real space $\R^d$ equipped with an arbitrary fixed norm $\|\cdot\|$ if not otherwise specified. For a set $A \subseteq \R^d$ we write $\lin(A)$, $\aff(A)$, $\conv(A)$, $\int(A)$, $\relint(A)$, and $\bd(A)$ for the linear, affine, or convex hull and the interior, relative interior and the boundary of $A$, respectively.

For two sets $A, B \subset \R^d$ and $\rho \in \R$, let $\rho A := \{\rho a: a \in A\}$ and $A+B:= \{a+b: a\in A, b\in B\}$ the $\rho$-\emph{dilatation}\index{dilatation} of $A$ and the \emph{Minkowski sum}\index{$A+B$ (for sets)}\index{Minkowski sum} of $A$ and $B$, respectively. 
We abbreviate  $A + (-B)$ by $A -B$ and  $A+\{c\}$ by $A+c$.

A non-empty set $K \subset \R^d$ which is convex and compact is called a \emph{convex body}\index{convex body} or \emph{body}\index{body} for short.
We write $\CC^d$\index{$\CC^d$} for the family of all convex bodies in $\R^d$ and $\CC^d_0$  if restricted to contain the origin as an interior point.

If a polytope $P \subseteq \R^d$ is described as a bounded intersection of halfspaces, we say that $P$ is in $\CH$-presentation\index{H-presentation@$\CH$-presentation}. If $P$ is given as the convex hull of finitely many points, we call this a $\CV$-presentation\index{V-presentation@$\CV$-presentation} of $P$.

For a convex set $K \subseteq \R^d$, $\ext(K)$ denotes  the set of  \emph{extreme points} of $K$.

We write $\B:=\{x \in \R^d: \|x\| \leq 1\}$ for the unit ball of $\|\cdot\|$ and  $\B_p:= \{x \in \R^d: \|x\|_p \leq 1\}$ for the unit ball of the $p$-norm $\|\cdot\|_p$. 

A set  $K \subseteq \R^d$ is called \emph{0-symmetric}\index{0-symmetric} if $-K= K$. If there is a $c \in \R^d$ such that $-(c+K)= c+K$ we call $K$ \emph{symmetric}\index{symmetric}.

For two vectors $x,y\in \R^d$, we use the notation $x^Ty:= \sum_{i=1}^d x_i y_i$\index{$x^Ty$}\index{scalar product}\index{dot product}\index{inner product} for the standard scalar product of $x$ and $y$, and by $H_{\leq}(a,\beta):= \{x \in \R^d: a^T x \leq \beta\}$ we denote the half-space\index{halfspace}\index{$H_\leq$} induced by $a\in \R^d$ and $\beta \in \R$, bounded by the hyperplane\index{hyperplane}\index{$H_=(a,\beta)$} $H_{=}(a,\beta):= \{x \in \R^d: a^T x = \beta\}$. 
For a vector $a \in \R^d$ and a convex set $K \subseteq \R^d$, we write
$ h(K,a):= \sup \{a^T x: x \in K\}$ for the \emph{support function} of $K$ in direction $a$.

For a convex function $f: C \rightarrow \R$ on some convex set $C\subseteq \R^d$ and $x \in C$, we write $\partial f(x)$ for the \emph{subdifferential}\index{subdifferential}\index{$\partial f(x)$} of $f$ in $x$ and, if $f$ is continuously differentiable in $x$, we denote by $\nabla f(x)$ the \emph{gradient}\index{$\nabla f(x)$} of $f$ in $x$.
We denote by $T(P,x):= \cl \{v \in \R^d: \exists \lambda > 0 ~s.t.~ x+\lambda v \in P\}$ the tangential cone  of $P$ at $x$ and by $N(P,x):= \{a \in \R^d: h(P,a)= a^T x\}= T(P,x)^\circ $
the normal cone\index{normal cone}\index{$N(P,x)$} of $P$ at $x$. 
For a convex body $C \in \CC^d$, we write $C^\circ:=\{a \in \R^d: a^T x \leq 1 ~\forall x \in C \}$ for its polar.

For $n\in \N$, we abbreviate $[n]:= \{1,\dots, n\}$\index{$[n]$}. 

\paragraph{} We denote by $\P$\index{$\P$} (and $\NP$\index{$\NP$}, respectively) the classes of decision problems that are solvable (verifiable, respectively) in polynomial time. For an account on complexity theory, we refer to \cite{GareyJohnson}. We write FPT\index{FPT}\index{$\text{FPT}$} for the class of fixed-parameter-tractable problems and W[1]\index{$W[1]$}\index{W[1]} for the problems of the first level of the W-hierarchy in the theory of Fixed Parameter Tractability. For an introduction to Fixed Parameter Tractability, we refer to the textbooks \cite{fg-06, n-06}.

\subsection{The Hausdorff Distance}

\paragraph{} We start by defining the functional of interest in this paper, which is based on the distance function of convex bodies.
\begin{defi}[Distance mapping]
For $P\subseteq \R^d$ non-empty and compact, define\index{$d(x, P)$}\index{nearest point mapping}
\begin{equation} 
\begin{array}{llll}
d(\cdot, P): 	&\R^d 	& \rightarrow 	& [0,\infty) \\
~ 				& x 	& \mapsto 		& d(x, P)= \min \left \{\|x-p\|: p \in P \right\}
\end{array}
\label{eq:distanceMapping}
\end{equation}
the distance mapping of $P$ induced by $\|\cdot\|$. As $P\neq \emptyset$ is compact and $\|\cdot\|$ is continuous, the notation in \eqref{eq:distanceMapping} as a \emph{minimum} is justified.
\end{defi}

\begin{rem}[Convexity of the distance function] \label{rem:distanceIsConvex}%
Since it is the key for the tractability result in Theorem~\ref{theo:HausdorffVVeasy}, we explicitly remark that the convexity of a set $P \subseteq \R^d$ directly implies the convexity of its distance function $d(\cdot, P)$.
\end{rem}

\begin{defi}[Hausdorff distance]
Let  $P,Q\subseteq \R^d$ non-empty and compact. The Hausdorff distance induced by $\|\cdot\|$ between $P$ and $Q$ is defined as
\begin{equation}
\delta(P,Q):= \max\{\max_{p \in P} d(p, Q), \max_{q\in Q} d(q, P)\}.
\label{eq:hausdorffDefi}
\end{equation}
Since $P, Q$ are non-empty and compact, and $d(\cdot, P)$, $d(\cdot, Q)$ are continuous, the maximum in \eqref{eq:hausdorffDefi} is attained. It can also be expressed by 
\begin{equation}
\delta(P, Q)= \min\{\rho \geq 0: P \subseteq Q + \rho \B, Q \subseteq P + \rho \B\}.
\label{eq:hausdorffDefi2}
\end{equation}
Figure~\ref{fig:defi} illustrates both formulations.

When working in $(\R^d, \|\cdot\|_p)$ for some $p \in \N\cup \{\infty\}$, we write 
$$\delta_p(P,Q) =\max\{\max_{x \in P} d_p(x, Q), \max_{q\in Q} d_p(q, P)\},$$
where the subscript $p$ explicitly indicates the norm with respect to which the Hausdorff distance is measured. 

\label{defi:hausdorff}
\end{defi}
\begin{figure}[]
\centering
\includegraphics[width=0.45\textwidth]{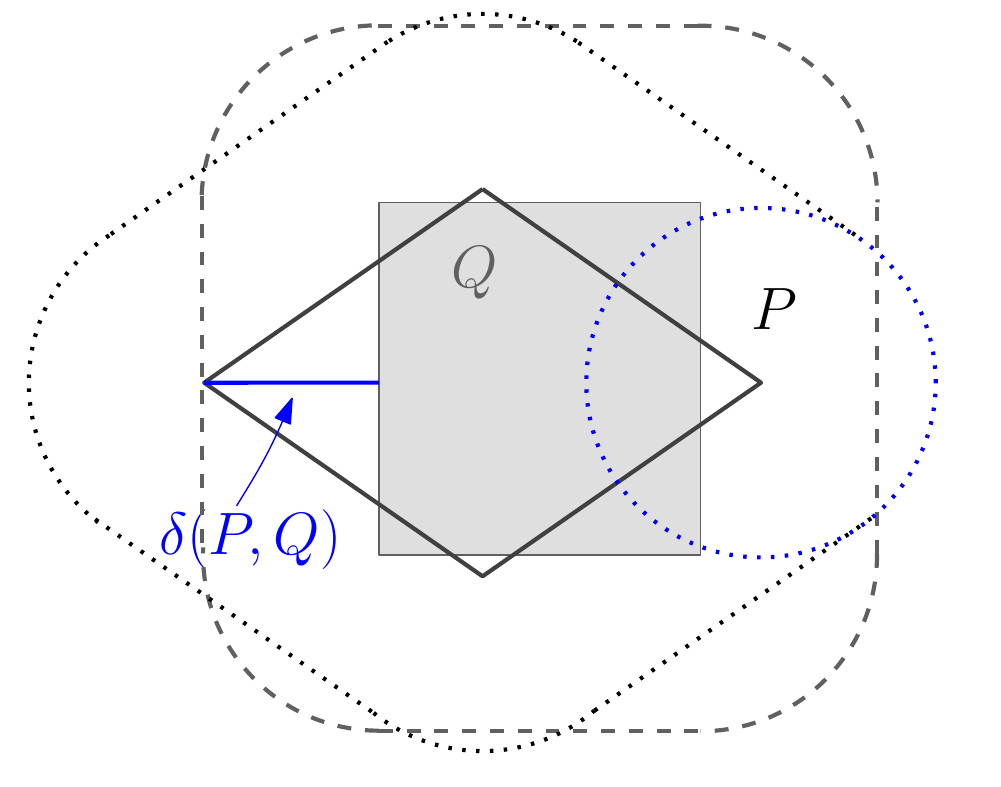}
\caption{Definition of the Hausdorff distance via the two equivalent formulations in Equations \eqref{eq:hausdorffDefi} and \eqref{eq:hausdorffDefi2} of Definition~\ref{defi:hausdorff}. \label{fig:defi}}
\end{figure}

\paragraph{} Besides the two equivalent formulations in \eqref{eq:hausdorffDefi} and \eqref{eq:hausdorffDefi2}, there is also a third formulation based on the support functions of $P$ and $Q$. 

\begin{lem}[Hausdorff distance via support functions] \label{lem:hausdorff3}%
Let $P,Q\in \CC^d$. Then, 
\begin{equation}
 \delta(P,Q)= \max_{u\in \B^\circ} |h(P, u) - h(Q, u)|.
 \label{eq:hausdorff3}
 \end{equation}
\end{lem}
\begin{proof}
Let $\rho:= \delta(P,Q)$ and $u \in \B^\circ$. As $P \subset Q + \rho \B$, it follows that $h(P,u) \leq h(Q+\rho \B, u) \leq h(Q, u) + \rho$. In the same way, $h(Q, u) \leq h(P, u) + \rho$. As $u \in \B^\circ$ was arbitrary, this yields $\max_{u\in \B^\circ} |h(P, u) - h(Q, u)| \leq \rho$.

For the other inequality, let $p^* \in P$ and $q^* \in Q$ such that $\rho= \|p^*-q^*\|$. Let $f(q):= \|q- p^*\|$. As $f(q^*) \leq f(q)$ for all $q \in Q$, there exists  
$u'\in \partial f(q^*)= \{u \in \B^\circ: u^T( q^*- p^*) = f(q^*)\}$ such that $-u' \in N(Q,q^*)$. Thus, for $u^* = - u'$, 
$$ \rho = f(q^*) = (u^*)^T (p^* - q^*) = (u^*)^T p^* - \max_{q\in Q} (u^*)^T q \leq h(P, u^*) - h(Q, u^*)  $$ 
$$\leq |h(P, u^*) - h(Q, u^*)| \leq \max_{u \in \B^\circ} |h(P, u) - h(Q, u)|.$$
\end{proof}

\paragraph{} Due to the homogeneity of  $h(P, \cdot), h(Q, \cdot)$ and $|\cdot|$, the maximum in \eqref{eq:hausdorff3} is attained for some vector $u^* \in \bd(\B^\circ)$. But since the function $f(u):=|h(P, u) - h(Q, u)|$ is not convex in general, the maximization in \eqref{eq:hausdorff3} cannot be restricted to $\ext(\B^\circ)$.

\section{The Evaluation Problem}
\label{sec:hausdorffEval}
\paragraph{} In this section, we study the complexity of computing the Hausdorff distance of two fixed polytopes. For the case of convex polygons in the plane, this question has first been studied in \cite{atallah-hausdorff}. The paper \cite{abb-amps-95} investigates the case of finite point sets and simple polygons in the plane. Both papers settle the issue by giving efficient algorithms that solve the respective problems. In \cite{alt-hausdorff}, these algorithms are extended in order to compute the Hausdorff distance between a union of $n$ and one of $m$ simplices of dimension $k$ in $\R^d$ in time $O(nm^{k+2})$.

Here, we study the computational complexity of the Hausdorff distance of two polytopes in arbitrary unbounded dimension. We investigate parameterized decision problems for a fixed $p$-norm and different presentations of the polytopes. Thus, throughout this section, let $p \in \N\cup \{\infty\}$ be fixed. 

The case where both input polytopes are given in $\CV$-presentation (\textsc{Hausdorff}$_p$-$\CV$-$\CV$) is stated explicitly as Problem~\ref{prb:hausdorffVV}. The case of two $\CH$-polytopes (\textsc{Hausdorff}$_p$-$\CH$-$\CH$) and the case of  mixed presentations (\textsc{Hausdorff}$_p$-$\CV$-$\CH$) are defined in the same way. In all three cases the dimension $d$ of the ambient space is part of the input and the parameter of the problem. For $p=\infty$, we interpret $\delta_\infty (P,Q)^\infty = \delta_\infty(P,Q)$.

\begin{prb}[\textsc{Hausdorff}$_p$-$\CV$-$\CV$]
\begin{tabular}{ll}
\textbf{Input:}		&  $d\in \N, n,m\in \N , p_1,\dots, p_n, q_1, \dots, q_m \in \Q^d, \rho \in \Q$ \\
\textbf{Parameter:} & $d$\\
\textbf{Question:} 	& Is $\delta_p \Bigl( \conv\{p_1,\dots, p_n\}, \conv\{q_1,\dots, q_m\} \Bigr)^p \geq \rho $?\\
\end{tabular}
\label{prb:hausdorffVV}
\end{prb}

\subsection{Tractability Results}
\paragraph{} In this subsection, we show that \textsc{Hausdorff}$_p$-$\CV$-$\CV$ can \emph{in principle} be solved efficiently. The restriction \enquote{in principle} is due to the fact, that for $p \in \N \setminus \{1,2\}$ (the $p$-th power of) the distance $d_p(x,P)^p$ of a point $x$ to a polytope $P$ might not be rational and therefore not computable in polynomial time. For $p\in \{1,2,\infty\}$, however, the distance to a rational polytope can be computed efficiently for either presentation of the polytope:

\begin{lem}[Computing $d(\cdot, P)$]
If the unit ball $\B\subseteq\R^d$ is a polytope given in $\CV$-presentation or in $\CH$-presentation,
then for any point $x\in \Q^d$ and any rational polytope $P\subseteq \R^d$ in $\CV$- or $\CH$-presentation, the distance
$d(x,P)$ can be computed in polynomial time.

For $\B= \B_2$, any point $x\in \Q^d$, and any rational polytope $P\subseteq \R^d$ in $\CV$- or $\CH$-presentation, $d_2(x,P)^2$ can be computed in polynomial time.
\label{lem:distLP}
\end{lem}
\begin{proof}
The cases where a $\CV$- or $\CH$-presentation of $\B$ is available can be solved by Linear Programming. We exemplarily show two combinations; the other two can be handled with the same techniques. 

If  $\B= \{x\in \R^d: u_j^T x \leq 1 ~\forall j\in [m]\}$ and  $P= \{x\in \R^d: Ax \leq b\}$, then
\begin{equation}
\begin{array}{llll}
d(x,P)= & \min 	& \rho \\
		& s.t.	&\rho \geq u_j^T (x - y)	& \forall j\in [m] \\
		&		& Ay \leq b \\
		&		& y\in \R^d \\
		& 		& \rho\in \R.
\end{array}
\label{eq:distVH}
\end{equation} 
If $\B= \conv\{v_1,\dots, v_m\}$ and $P= \conv\{p_1,\dots, p_n\}$, then
\begin{equation}
\begin{array}{llll}
d(x,P)=	& \min	&\rho \\
		&s.t.	& x - \sum_{i=1}^n \lambda_i p_i = \sum_{j=1}^m \mu_j v_j	\\
		&		& \sum_{i=1}^n \lambda_i= 1 \\
		&		& \sum_{j=1}^m \mu_j= \rho\\
		&		& \mu_j \geq 0 	& \forall j\in [m] \\
		&		& \lambda_i \geq 0 	& \forall i\in [n] \\
		&		& \rho \geq 0.
\end{array}
\label{eq:distVV}
\end{equation}

Now, let $\B= \B_2$. If $P= \{x\in \R^d: Ax \leq b\}$, then
\begin{equation}
\begin{array}{llll}
d_2(x,P)^2=	& \min	&(p-x)^T (p-x) \\
		&s.t.	& Ap \leq b.	\\
		&		& p\in \R^d.
\end{array}
\label{eq:distanceQuadProgH}
\end{equation}
By \cite{khach-79}, the optimal solution of the convex quadratic program with linear constraints \eqref{eq:distanceQuadProgH} %and \eqref{eq:distanceQuadProgV} 
is rational and can be found in polynomial time. Again, the case where $P$ is given in $\CV$-presentation can be handled with the same formulation as in \eqref{eq:distVV}
\end{proof}

\paragraph{} Together with the convexity of the distance function (Remark~\ref{rem:distanceIsConvex}), Lemma~\ref{lem:distLP} implies the following:

\begin{theo}[Tractability of \textsc{Hausdorff}$_p$-$\CV$-$\CV$]
If a $\CV$- or $\CH$-presentation of the unit ball $\B$ can be computed in polynomial time, the Hausdorff distance of two rational $\CV$-polytopes $P,Q \subseteq \R^d$  can be computed in polynomial time. 

In particular, \textsc{Hausdorff}$_1$-$\CV$-$\CV$, \textsc{Hausdorff}$_2$-$\CV$-$\CV$ and \textsc{Hausdorff}$_\infty$-$\CV$-$\CV$ are in $\P$.
\label{theo:HausdorffVVeasy}
\end{theo}
\begin{proof}
Let $P:= \conv\{p_1,\dots, p_n\}$ and $Q:= \conv\{q_1, \dots, q_m\}$ be rational. Since $d(\cdot, Q)$ is convex, $\max_{x\in P} d(x, Q)$ is attained at 
a vertex of $P$ as well as $\max_{x\in Q} d(x, P)$ is attained at a vertex of $Q$. Hence,
$\delta(P,Q)= \max \bigl\{\max\{d(p, Q): p \in P\}, \max\{d(q, P): q\in Q\}\bigr\}$ can be computed by determining  $d(p_i, Q)$ and $d(q_j,P)$ for all $i\in [n]$ and $j\in [m]$.  These values can be computed efficiently by Lemma~\ref{lem:distLP}.
\end{proof}

\paragraph{} Using the argument from the proof of Theorem~\ref{theo:HausdorffVVeasy}, we can at least state the following for general $p \in \N$.

\begin{rem}[Approximation for \textsc{Hausdorff}$_p$-$\CV$-$\CV$]
For $p \in \N$ and rational polytopes $P= \conv \{p_1,\dots, p_n\}\subseteq \R^d$ and $Q= \conv\{q_1,\dots, q_m\}\subseteq \R^d$, the Hausdorff distance $\delta_p(P,Q)$ can be approximated to any accuracy in polynomial time by using the Ellipsoid Method \cite{gls-93} for the approximation of $d(p_i,Q)$ and $d(q_j, P)$ for $i\in [n]$ and $j \in [m]$.
\end{rem}

\begin{rem}[Direct approximation of $\delta_2(P,Q)$]
For two rational $\CV$-polytopes $P,Q \subseteq \R^d$, the Hausdorff distance $\delta_2(P,Q)$ can also be approximated to any accuracy in polynomial time by solving the Second Order Cone Program which arises from directly spelling out the definition of the Hausdorff distance in \eqref{eq:hausdorffDefi2}, cf. the SOCP in the proof of Theorem~\ref{theo:matchingSOCP} with the restrictions $\alpha=1$ and $c= 0$.

\end{rem}

\subsection{Hardness Results}
\paragraph{} We complement the results of the previous subsection by showing that in almost all other cases it is $\NP$-hard and W[1]-hard to bound the Hausdorff distance of two polytopes. For this purpose, we apply the general reduction technique described in \cite{christianFPTintro}. Thus, for some $k \in \N$, we will deal with polytopes in $\R^{2k} $, which we consider as
$$\R^{2k}= \R^2 \times \R^2 \times\dots \times\R^2,$$
i.e. we will think of a vector $x\in \R^{2k}$ as $k$ two-dimensional vectors stacked upon each other. In this setting, the following notation is convenient.

\begin{nota}\label{nota:halfspaceR2}%
By indexing a vector $x \in \R^{2k}$, we always refer to the $k$ two-dimensional vectors $x_1,\dots, x_k \in \R^2$ such that $x= (x_1^T, \dots, x_k^T)^T$. 
Further, for $a \in \R^2$ and $\beta \in \R$, we let\index{$H_\leq^i(a, \beta)$}
$$H_\leq^i(a, \beta):= \{x \in \R^{2k}: a^T x_i \leq \beta\}.$$
\end{nota}

\paragraph{} The following technical remark enables the reduction technique of \cite{christianFPTintro} to be applied in the proof of Lemma~\ref{lem:hausdorff1HW1}. 

\begin{lem}[Hausdorff distance of certain direct products] \label{lem:hausdorffDirectProd}%
For $k \in \N$, let $P_1,\dots, P_k \subseteq \R^2$ and $Q_1,\dots, Q_k \subseteq \R^2$ be polytopes with $Q_i \subseteq P_i$ for all $i \in [k]$ and further $P:= P_1\times \dots \times P_k \subseteq \R^{2k}$ and $Q:= Q_1\times \dots \times Q_k\subseteq \R^{2k}$. Then, for all $p\in \N$,
$$ \delta_p(P,Q)^p = \sum_{i=1}^k \delta_p(P_i, Q_i)^p.$$
\end{lem}
\begin{proof}
Since $Q_i \subseteq P_i$ for all $i \in [k]$ and the definitions of $Q$ and $P$ as cartesian products, we have $\delta_p(P,Q)^p = \max_{x \in P} \sum_{i=1}^k \delta_p(x_i, Q_i)^p = \sum_{i=1}^k \max_{x_i \in P_i} \delta_p(x_i, Q_i)^p = \sum_{i=1}^k \delta_p(P_i, Q_i)^p$.

\end{proof}

\paragraph{} For the first hardness proof, we will reduce the W[1]-complete problem \textsc{Clique} to  \textsc{Hausdorff}$_1$-$\CH$-$\CH$. The formal parametrized decision problem of \textsc{Clique} is given in Problem~\ref{prb:clique}; a proof of its W[1]-completeness can be found e.g.~in \cite[Theorem 6.1]{fg-06}. 
Moreover, it is shown in \cite{Chen2005216} that \textsc{Clique}  cannot be solved in time $n^{o(k)}$, unless the Exponential Time Hypothesis\footnote{The Exponential Time Hypothesis conjectures that n-variable 3-CNFSAT cannot be solved in $2^{o(n)}$-time; cf. \cite{ip-01}.} fails. 

\begin{prb}[\textsc{Clique}] \label{prb:clique}%
\begin{tabular}{ll}
\textbf{Input:} &$n,k \in \N$, $E \subseteq \left({[n]} \atop 2 \right)$ \\
\textbf{Parameter:} & $k$ \\
\textbf{Question:}& Does $G=([n], E)$ contain a clique of size $k$?
\end{tabular}
\end{prb}

\begin{lem}[{W[1]-Hardness of \textsc{Hausdorff}$_1$-$\CH$-$\CH$}]
\textsc{Hausdorff}$_1$-$\CH$-$\CH$ is W[1]-hard, even if restricted to 0-symmetric polytopes.
\label{lem:hausdorff1HW1}
\end{lem}
\begin{proof}
Let $(m,k,E)$ be an instance of  \textsc{Clique} with $m$ vertices. Define $n:= 2m$ and let  $p_1',\dots, p_{2n}' \in \S_2$ be the vertices of a regular $2n$-gon in the Euclidean unit circle in the plane.  Let  $U:= \frac{1}{n^{2p}k^2}$ and, for $v \in [n]$, let $\bar p_v$ be the rounding of $p_v'$ to the grid $\frac{U}{2} \Z^2$ and define $\bar p_v = -\bar p_{v-n}$ for $v\in [2n]\setminus[n]$.

As demonstrated in \cite[Lemma 2.6]{normmaxW1}, it is easy to verify that for
$$P_1:= \{x \in \R^2: \bar p_v^T x \leq 1 ~\forall v\in [2n]\} $$ 
each inequality $\bar p_v^T x \leq 1 $ induces a facet of $P_1$. As well, the coding length of the above $\CH$-presentation of $P_1$ is polynomial in $m$ and $k$. By relabeling and scaling, we can achieve that

$$P_1 = \bigcap_{v\in [n]} \bigl(H_\leq(a_v ,\beta_v) \cap H_\leq(c_v, \gamma_v)\bigr), $$ 
where the facets induced by $a_v$ and $c_v$ alternate along the boundary of $P_1$, $\|a_v\|_\infty =\|c_v\|_\infty =1$ and $\beta_v,\gamma_v> 0$ for all $v\in [n]$, cf. Figure~\ref{fig:alternatingFacets}. Observe that the sets $\{c_1,\dots, c_n\}$ and $\{a_1,\dots, a_n\}$ are 0-symmetric.

Now, compute all vertices of $P_1$ in time $O(n\log(n))$ (cf. e.g.~\cite{CGdeBerg}), such that $P_1= \conv\{p_1,\dots, p_{2n}\}$ and let, for $v \in [n]$, 
$\varepsilon_v:= \gamma_v - \max \{c_v^T p_w: c_v^T p_w \neq \gamma_v, ~ w \in [2n]\}>0,$ $\varepsilon:= \frac{1}{10} \min \{\varepsilon_v: v \in [n]\}>0$,
 and 
$$ Q_1 := \{x \in \R^2: a_v^T x \leq \beta_v,~c_v^T x \leq \gamma_v -\varepsilon ~\forall v \in [n]\}.$$

\begin{figure}[h]
\centering
\includegraphics[width=0.6\textwidth]{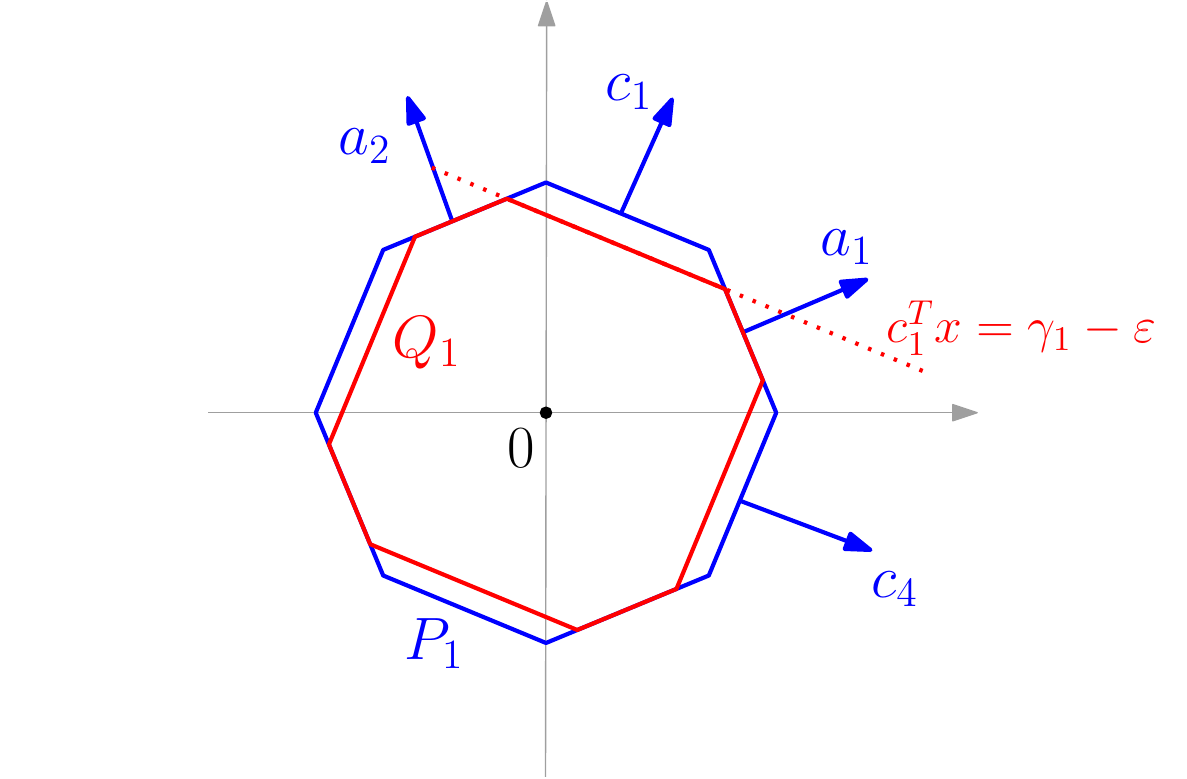}
\caption[Polytopes in the reduction of \textsc{Clique} to \textsc{Hausdorff}$_1$-$\CH$-$\CH$]{Illustration of the two polytopes $P_1,Q_1 \subseteq \R^2$ from the reduction in the proof of Lemma~\ref{lem:hausdorff1HW1}.\label{fig:alternatingFacets}}
\end{figure}

For arbitrary  $u \in \bd(\B_\infty^2)$, let $p(u) \in P_1$ and $q(u) \in Q_1$ be vertices of the two polytopes with $u \in N(P_1, p(u))= N(Q_1, q(u))= \pos\{c_v, a_w\}$ for some $v \in [n]$ and $w \in \{v-1, v\}$. Then, we can express $u= \lambda c_v + \mu a_w$ with $\lambda, \mu\geq 0$. Since $\|c_v\|_\infty = 1$, there exists $e \in \{\pm e_1,\pm e_2\}$ such that $e^T c_v=1$. We can assume without loss of generality that $n$ is sufficiently large such that $e^T a_w >0$, and thus,
$1 \geq e^T u = \lambda{e^T c_v} + \mu {e^T a_w} \geq  \lambda$,
where equality holds if and only if $\lambda =1$ and $\mu =0$.

Therefore, $h(P_1,u) - h(Q_1,u)=  \lambda c_v^T(p(u)-q(u)) + \mu a_w^T (p(u)-q(u)) = \lambda \varepsilon \in [0, \varepsilon]$, and hence, by Lemma~\ref{lem:hausdorff3},
\begin{equation}
\delta_1(P_1, Q_1)=  \max_{u\in \B_\infty^2}(h(P_1, u) - h(Q_1, u)) = \varepsilon,
\label{eq:hausdorffR2}
\end{equation}
where the maximum in \eqref{eq:hausdorffR2} is attained if and only if $u \in \{c_1,\dots, c_{n}\}$.

Using Notation~\ref{nota:halfspaceR2}, let
$$P_2:= \bigcap_{i\in [k]} \bigcap_{v \in [n]} \bigl(H_\leq^i(a_v, \beta_v)\cap H_\leq^i (c_v, \gamma_v)\bigr) \subseteq \R^{2k}$$ 
 and 
$$ Q:= \bigcap_{i\in [k]} \bigcap_{v \in [n]} \bigl(H_\leq^i(a_v, \beta_v)\cap H_\leq^i (c_v, \gamma_v-\varepsilon)\bigr) \subseteq \R^{2k}.$$ 

By Lemma~\ref{lem:hausdorffDirectProd} and \eqref{eq:hausdorffR2}, we obtain $\delta_1(P_2,Q) = k\varepsilon$.

For $i,j \in [k]$ and $v,w \in [m]$ define
$$E^{ij}_{vw}:= \{x \in \R^{2k}:-(\gamma_v + \gamma_w - \varepsilon)\leq  c_v^T x_i + c_w^T x_j \leq \gamma_v + \gamma_w - \varepsilon\},$$
$$F^{ij}_{vw}:= \{x \in \R^{2k}:-(\gamma_v + \gamma_w - \varepsilon)\leq  c_v^T x_i - c_w^T x_j \leq \gamma_v + \gamma_w - \varepsilon\}$$
%$$F^{ij}_{v}:= \{x \in \R^{2k}: c_v^T x_i + c_v^T x_j \leq 2\gamma_v  - \varepsilon\}$$
and, for $N:= \binom{[m]}{2}\setminus E$, let 
$$P:= P_2 \cap \bigcap_{{\{v,w\}\in N}\atop{{i,j \in [k], i\neq j}}} (E^{ij}_{vw}\cap F^{ij}_{vw}) \cap \bigcap_{{v \in [m]}\atop{{i,j \in [k], i\neq j}}} (E^{ij}_{vv}\cap F^{ij}_{vv})$$

\paragraph{} Now, by definition, $P$ and $Q$ are 0-symmetric and we claim that $\delta_1(P,Q)= k\varepsilon$ if and only if $G=([m], E)$ has a clique of size $k$. Since \textsc{Clique} is W[1]-complete \cite[Theorem~6.1]{fg-06}, this completes the hardness proof for \textsc{Hausdorff}$_1$-$\CH$-$\CH$.

Assume $\delta_1(P,Q)= k\varepsilon$. Since $Q \subseteq P$, there exists $u^* \in \B_\infty^{2k}$ such that $h(P, u^*) = h(Q, u^*) + k\varepsilon$. Since $P \subseteq P_2$, we obtain via the sharpness condition in \eqref{eq:hausdorffR2}, that $u^* = ( c_{v_1}^T, \dots, c_{v_k}^T)^T$ for some $v_1,\dots,  v_k \in [n]$. Hence, there is a vector $x^* \in P$ such that 
$$ c_{v_1}^T x_1^* = \gamma_{v_1}, \quad \dots, \quad c_{v_k}^T x_k^*= \gamma_{v_k}.$$
If $c_{v_j} = \pm c_{v_i}$ for some $i\neq j \in [k]$, this would imply $c_{v_i}^T x_i^* \pm c_{v_i}^T x_j^*= 2\gamma_{v_i}$ which contradicts $x^* \in P \subseteq E_{v_iv_i}^{ij} \cap F_{v_iv_i}^{ij}$. Hence, defining for $i\in [k]$,

$$u_i:= \left\{\begin{array}{ll} v_i &\text{ if } v_i \leq m \\ v_i -m &\text{ else} \end{array}\right.$$
yields $|\{u_1,\dots, u_k\}|=k$. In the same way, $c_{v_i}^T x_i^* + c_{v_j}^T x_j^* = \gamma_{v_i} +\gamma_{v_j}$ implies that $\{u_i,u_j\} \in E$. Thus, $\{u_1,\dots, u_k\}$ is the set of vertices of a clique of size $k$ in $G$.

If on the other hand $\{v_1,\dots, v_k\}\subseteq [m]$ is the vertex set of a $k$-clique, then $u^*:= (c_{v_1}^T, \dots, c_{v_k}^T)^T $ satisfies $h(P, u^*)= h(Q, u^*) + k\varepsilon$ and therefore $\delta_1(P,Q)=k\varepsilon$.
\end{proof}

\paragraph{} For all remaining values $p\in \N$, hardness can be established by a much simpler reduction of the following problem which is W[1]-hard for all $p \in \N\setminus\{1\}$, cf. \cite{normmaxW1}.

\begin{prb}[{\normmax[p]}] \label{prb:normmax}% 
\begin{tabular}{ll}
\textbf{Input:} &$d \in \N$, $\gamma \in \Q$, rational $\CH$-presentation of a 0-symmetric polytope $P \subseteq \R^d$ \\
\textbf{Parameter:} & $d$ \\
\textbf{Question:}& Is $\max \{\|x\|_p^p: x \in P \} \geq \gamma $?
\end{tabular}
\end{prb}

\begin{lem}[{Reduction of \normmax$_p$}]
For all $p \in \N \cup \{\infty\}$, \normmax[p] can be reduced in polynomial and FPT (with respect to the dimension) time to \textsc{Hausdorff}$_p$-$\CH$-$\CH$ and \textsc{Hausdorff}$_p$-$\CV$-$\CH$
\label{lem:NormmaxToHausdorff}
\end{lem}
\begin{proof}
If $(d, P, \lambda)$ is an instance of \textsc{Normmax}$_p$  with an $\CH$-presented rational polytope $P\subseteq \R^d$, let $Q:= \{x \in \R^d: \pm e_i^T x \leq 0 ~\forall i\in [d]\} = \conv\{0\} = \{0\}$. Then,
$$\max \{\|x\|_p^p : x \in P\} \geq \lambda \Leftrightarrow \delta_p(P,Q)^p \geq \lambda$$
\end{proof}

\begin{theo}[Hardness of \textsc{Hausdorff}$_p$-$\CH$-$\CH$ and \textsc{Hausdorff}$_p$-$\CV$-$\CH$]
For $p \in \N$, \textsc{Hausdorff}$_p$-$\CH$-$\CH$ is W[1]-hard. For $p\in \N\setminus\{1\}$, \textsc{Hausdorff}$_p$-$\CV$-$\CH$ is W[1]-hard. For $p \in \N$, \textsc{Hausdorff}$_p$-$\CH$-$\CH$ and \textsc{Hausdorff}$_p$-$\CV$-$\CH$ are $\NP$-hard. All statements hold true even if all polytopes are restricted to be 0-symmetric.
  
\label{theo:HausdorffHVhard}
\label{theo:HausdorffHHard}
\end{theo}
\begin{proof}
By Lemma~\ref{lem:NormmaxToHausdorff} and the W[1]-hardness of \normmax[p] for $p \in \N\setminus\{1\}$ from \cite{normmaxW1}, \textsc{Hausdorff}$_p$-$\CH$-$\CH$ and \textsc{Hausdorff}$_p$-$\CV$-$\CH$ are W[1]-hard for $p\in \N\setminus\{1\}$. W[1]-hardness of \textsc{Hausdorff}$_1$-$\CH$-$\CH$ follows from Lemma~\ref{lem:hausdorff1HW1}.
Since \normmax[p] is $\NP$-hard for all $p \in \N$ (e.g.~\cite{bgkl-90, ms-86}), the same reduction can also be used to show $\NP$-hardness for all $p \in \N$. 
\end{proof}

\paragraph{} We point out that the complexity of \textsc{Hausdorff}$_\infty$-$\CH$-$\CH$ and \textsc{Hausdorff}$_\infty$-$\CV$-$\CH$ is left open by Theorem~\ref{theo:HausdorffHHard}. For the first problem, we give a partial answer in Corollary~\ref{cor:HHpolytopal}. The latter one has an interesting connection to another decision problem (Problem~\ref{prb:vertexEnumeration}) whose complexity is still unknown.

\paragraph{} As a corollary of Lemma~\ref{lem:hausdorff1HW1}, we obtain the hardness of computing the Hausdorff distance of two $\CH$-polytopes in an arbitrary polytopal norm which is part of the input. Hence, in contrast to the results about norm maximization in \cite{normmaxW1}, the approximation of the unit ball of a $p$-norm by a polytope cannot be used as a polynomial time approximation algorithm for the Hausdorff distance of two $\CH$-presented polytopes.

\begin{cor}[{W[1]-hardness for polytopal norms}]
For 0-symmetric $\B\in \CC_0^d$ and $P,Q \in \CC^d$, let $\delta_\B(P,Q)$ denote the Hausdorff distance induced by the norm with unit ball $\B$. Then the problems

\medskip
\begin{tabular}{ll}
\textbf{Input:}		&  $d\in \N$, $\rho \in \Q$, $P,Q\subseteq \R^d$ rational polytopes in $\CH$-presentation,\\ 
& polytopal unit ball $\B \subseteq \R^d$ in $\CH$-presentation ($\CV$-presentation, resp.)\\
\textbf{Parameter:} & $d$\\
\textbf{Question:} 	& Is $\delta_\B \left( P, Q\right) \geq \rho $?\\
\end{tabular}
\medskip

are both W[1]-hard.
\label{cor:HHpolytopal}
\end{cor}
\begin{proof}
Since we can construct an $\CH$- and $\CV$-presentation of $\B_1$ in FPT-time, we can reduce  \textsc{Hausdorff}$_1$-$\CH$-$\CH$ to either of the above problems.
\end{proof}

\paragraph{} Regardless of the norm used to measure the Hausdorff distance, \textsc{Hausdorff}-$\CV$-$\CH$ is closely linked to (and at least as hard as) the following vertex enumeration problem.

\begin{prb}[\textsc{VertexEnumeration}]
\begin{tabular}{ll}
\textbf{Input:}		&  $d\in \N, n,m\in \N , a_1, \dots, a_m \in \Q^d, \beta_1,\dots, \beta_m\in\Q,  p_1,\dots, p_n \in \Q^d$ \\
\textbf{Question:} 	&  Is $\ext\{x\in \R^d: a_i^Tx \leq \beta_i ~\forall i \in [m] \} \setminus \{p_1,\dots,p_n\} \neq \emptyset$?\\
\end{tabular}
\label{prb:vertexEnumeration}
\end{prb}

\paragraph{} Clearly, if $Q:= \{x \in \R^d: a_i^T x \leq \beta_i ~\forall i \in [m]\}$ is bounded and $P:= \conv\{p_1,\dots, p_n\}$, then the question in \textsc{VertexEnumeration} is equivalent to the question, whether $\delta(P,Q)> 0$ in any norm. 
In \cite{khachiyanVertexEnumeration}, \textsc{VertexEnumeration} is shown to be $\NP$-hard for a general (possibly unbounded) polyhedron $Q$.  The question for the complexity of the same decision problem where $Q$ is required to be bounded is open \cite{khachiyanVertexEnumeration} and we refer to \cite{hansDiss} for a recent account on this topic and related results.

\section{The Matching Problem}
\label{sec:hausdorffMatching}
\paragraph{} The remaining part of the paper is concerned with the Hausdorff matching problem under homothetics for convex bodies. 

 In order to avoid degeneracies, we assume both bodies to have an interior point (the origin without loss of generality). Since it is probably the case of most interest, we state the problem and the results only for the Euclidean norm, which also offers some notational convenience.  Arbitrary norms can be handled with the same ideas at the expense of a slightly bulkier notation.

\begin{prb}[Hausdorff matching under homothetics]
For $P, Q\in \CC_0^d$, find a scaled translate of $P$ such that its Hausdorff distance to $Q$ is minimized.\index{Hausdorff matching} In other words, solve the problem\index{homothetic transformation}

\begin{equation} 
\begin{array}{rlll}
\delta_H(P,Q):=	&\min & \delta_2(\alpha P + c, Q) \\
				&s.t. & c \in \R^d \\
				&	& \alpha > 0
\end{array}
\label{eq:hausdorffMatchingGeneral}
\end{equation}
with $\delta_H(P,Q)$ being the optimal value of the optimization problem in \eqref{eq:hausdorffMatchingGeneral}. Here, the subscript \enquote{$H$} indicates that the Hausdorff distance (in Euclidean norm) of $P$ and $Q$ is measured \enquote{up to homothetics}. Moreover, we say that $P$ is in \emph{optimal homothetic position}\index{optimal homothetic position} with respect to $Q$ if 
\begin{equation}
\delta_H(P, Q)= \delta(P,Q) .
\end{equation}
\label{prb:HMHomo}
\end{prb}

\subsection{Optimality Criterion}
\label{sec:optCondHausdorff}
\paragraph{} The goal of this subsection is to give an optimality criterion for Problem~\ref{prb:HMHomo} and to characterize when $P$ is in optimal homothetic position with respect to $Q$ in the spirit of John's Theorem \cite{john}.
 
\paragraph{} In preparation of Theorem~\ref{theo:optCondHausdorffHomo}, we first prove a series of technical lemmas that provide required details such as normal cones of outer parallel bodies, several derivatives and a statement about subgradients of certain convex functions. 
We start by recalling that, for $P \in \CC^d$ and $x \in \R^d$, there is a unique point $\Pi_P(x) \in P$ such that 
$$d_2(x, P) = \|x - \Pi_P(x)\|_2.$$
Thus, the nearest-point-mapping $ \Pi_P: \R^d \mapsto P;~ x \mapsto \Pi_P(x)$ is well-defined.

\begin{lem}[Normal cone of outer parallel bodies] \label{lem:outerParallel}%
Let $P \in \CC^d$ and $\rho > 0$. For $x \in \bd(P+\rho \B_2)$, the normal cone of $P+\rho \B_2$ in $x$ is given by  
$$N(P+ \rho \B_2, x) =\pos\bigl\{ x - \Pi_P(x)\bigr\}.$$
If further $\rho' >0$, then 
$$ P+ \rho' \B_2 = \bigcap_{x \in \bd(P+\rho \B_2)}H_\leq\Bigl(x-\Pi_P(x),~ (x-\Pi_P(x))^T \Pi_P(x) + \rho\rho'\Bigr).$$
\end{lem}
\begin{proof}
Since $\Pi_P(x) + \rho \B_2 \subseteq P +\rho \B_2$, we have $H_\leq( x- \Pi_P(x), 0) =T(\Pi_P(x) + \rho \B_2,x) \subseteq T(P+\rho \B_2, x)$.
Since $P+ \rho \B_2$ is convex and the tangential cone $T(P+ \rho \B_2,x)$ is thus at most a half-space, we conclude that the above inclusion is actually an equality and that  $N(P +\rho \B_2, x)= T(\Pi_P(x) + \rho \B_2, x)^\circ = \pos\{x - \Pi_P(x)\}.$

For $\rho'>0$, the above yields
\begin{align*} 
P+ \rho' \B_2 &=&& \bigcap_{x' \in \bd(P+\rho'\B_2)} H_\leq\Bigl(x'-\Pi_P(x'), \bigl(x'-\Pi_P(x')\bigr)^Tx'\Bigr)\\
%&=&& \bigcap_{p \in \bd(P+\rho\B_2)} H_\leq\Bigl(p-\Pi_P(p), \bigl(p-\Pi_P(p)\bigr)^Tp'\Bigr)\\
&=&&\bigcap_{x \in \bd(P+\rho\B_2)} H_\leq\Bigl(x-\Pi_P(x), \bigl(x-\Pi_P(x)\bigr)^T\left(\Pi_P(x) + \frac{\rho'}{\rho}(x-\Pi_P(x))\right)\Bigr)\\
&=&&\bigcap_{x \in \bd(P+\rho\B_2)} H_\leq\Bigl(x-\Pi_P(x),\bigl(x-\Pi_P(x)\bigr)^T \Pi_P(x) + \rho\rho' \Bigr).
\end{align*}
\end{proof}

\begin{lem}[Differentiating $d_2(\cdot, P)$]\label{lem:nablaGp}%
Let $P \in \CC^d$ and define $ g_P: \R^d \rightarrow [0, \infty); x \mapsto d_2(x, P)$. %= \min \{\|x-p\|_2: p \in P\}.$
The mapping $g_P$ is convex and, for $x \in \R^d \setminus P$, it is continuously differentiable in $x$ with
$$\nabla g_P(x) = \frac{x-\Pi_P(x)}{\|x-\Pi_P(x)\|_2}.$$
\end{lem}
\begin{proof}
Convexity of $g_P$ was already observed in Remark~\ref{rem:distanceIsConvex}. Since, for $x \in \R^d\setminus P$, the function $g_P$ does not attain its minimum in $x$, Theorem 23.7 from \cite{rockafellar} asserts that
$\pos(\partial g_P(x))= N(P + d_2(x,P)\B_2, x).$
 Hence, by Lemma~\ref{lem:outerParallel}, $\partial g(x) = \{\lambda (x - \Pi_P(x))\}$ for some $\lambda >0$. Using the linearity of  $\|\cdot\|_2$ on rays emanating from the origin, we obtain $\lambda = \|x-\Pi_P(x)\|_2^{-1}$.
As $|\partial g(x)|=1$ for all $x\in \R^n \setminus P$, Theorem 25.1 and Corollary 25.5.1 in \cite{rockafellar} yield that $g$ is continuously differentiable in $x$ and that $\nabla g_P(x) = {\|x-\Pi_P(x)\|_2}^{-1}(x-\Pi_P(x))$.
\end{proof}

\begin{lem}[Differentiating the objective function]\label{lem:fpfq}%
Let $P,Q \in \CC^d$ and define, for $p\in P$, $ f_p: (0,\infty) \times \R^d \rightarrow [0, \infty); (\alpha, c)\mapsto d_2(\alpha p + c, Q)$
and, for $q \in Q$, $ f_q: (0,\infty) \times \R^d \rightarrow [0, \infty); (\alpha, c)\mapsto d_2(q, \alpha P +c).$
Then, $f_p$ and $f_q$ are convex.
Further, for $p \in \R^d \setminus Q$,
\begin{equation}
\nabla f_p(1,0)= \frac{1}{\|p-\Pi_Q(p)\|_2}\vector{(p-\Pi_Q(p))^T p \\ p-\Pi_Q(p) }
\label{eq:nablaFp}
\end{equation}
and for  $q \in \R^d \setminus P$,
\begin{equation}
\nabla f_q(1,0)= -\frac{1}{\|q-\Pi_P(q)\|_2}\vector{(q-\Pi_P(q))^T \Pi_P(q) \\ q-\Pi_P(q) }
\label{eq:nablaFq}
\end{equation}
\end{lem}
\begin{proof}
For $p \in P$, the function $f_p$ is a composition of a linear function and a convex function and hence convex. 

%For $x,y\in \R^d$ and $K,L \in \CC^d$, we have $d(x+y, K+L) \leq d(x,K) + d(y, L),$ as $\Pi_K(x) +\Pi_L(y) \in K+L$. 

Now, let  $q\in Q$. For the convexity of $f_q$, let $\alpha_1,\alpha_2 >0$, $c_1,c_2 \in \R^d$ and $\lambda \in [0,1]$. Using  $d(x+y, K+L) \leq d(x,K) + d(y, L)$, we obtain $ f_q\bigl(\lambda \alpha_1 + (1-\lambda)\alpha_2, \lambda c_1 + (1-\lambda)c_2\bigr)= d\bigl( \lambda q +(1-\lambda)q, \lambda(\alpha_1 P +c_1) + (1-\lambda)(\alpha_2 P + c_2)\bigr) \leq d\bigl(\lambda q, \lambda(\alpha_1P +c_1)\bigr) + d\bigl((1-\lambda)q, (1-\lambda)(\alpha_2 P + c_2)\bigr) = \lambda f_q(\alpha_1, c_1) + (1-\lambda)f_q(\alpha_2, c_2),$
which shows the convexity of $f_q$.

A direct application of the chain rule together with Lemma~\ref{lem:nablaGp} yields \eqref{eq:nablaFp}. In order to differentiate $f_q$, we express $f_q(\alpha, c)= d(q, \alpha P +c)= \alpha g_P(\frac{1}{\alpha}(q-c))$ with $g_P$ defined as in Lemma~\ref{lem:nablaGp}. Differentiating this expression yields

$$ \frac{\partial f_q}{\partial c} (\alpha, c)= -\alpha \nabla g_P\left(\frac{1}{\alpha} (q-c)\right)$$
and 
$$ \frac{\partial f_q}{\partial \alpha}(\alpha, c) = g_P\left(\frac{1}{\alpha}(q-c)\right) -\frac{1}{\alpha} \nabla g_P\left(\frac{1}{\alpha}(q-c)\right)^T (q-c).$$

Plugging in $\alpha=1$ and $c=0$ and using $g_P(q)^2 = (q- \Pi_P(q))^T (q- \Pi_P(q))$, we obtain \eqref{eq:nablaFq}.
\end{proof}

\paragraph{} As a consequence of the convexity of $f_p$ and $f_q$ for all $p \in P$ and $q \in Q$ and the definition of the Hausdorff distance in \eqref{eq:hausdorffDefi}, we see that the objective function in the computation of $\delta_H(P,Q)$ is convex so that a necessary and sufficient optimality condition can be expected. The next lemma now investigates the subdifferential of the objective function, which is a supremum of uncountably many convex functions.

\begin{lem}[Subgradient of the supremum of convex functions] \label{lem:subgradientInclusion}%
Let $I$ be a (possibly uncountable) index set and $f_i: \R^d \rightarrow \R$ convex for $i \in I$. Let $f: \R^d \rightarrow \R; f(x):= \sup \{f_i(x): i\in I\}$
and for $x \in \R^d$ define $\CA(x):= \{i \in I: f_i(x)= f(x)\}$. Then, for all $x \in \R^d$,
\begin{equation}
 \cl\Bigl( \conv \Bigl( \bigcup_{i \in \CA(x)} \partial f_i(x)   \Bigr) \Bigr)  \subseteq \partial f(x). 
 \label{eq:subgradientInclusion}
\end{equation}
\end{lem}
\begin{proof}
Let $x \in \R^d$, $i \in \CA(x)$ and $a \in \partial f_i(x)$. Let further $\bar a := \left( a \atop {-1} \right)$ and $\beta:= \bar a ^T \left(x \atop {f(x)}\right)$. Then, $H_=(\bar a,\beta)$ supports $\epi(f_i)$ and $\epi(f_i) \subseteq H_\leq(\bar a, \beta)$ (cf. \cite[Section 23]{rockafellar}).

Since $\epi(f)= \bigcap_{i\in I} \epi(f_i)$, we have $\epi(f) \subseteq H_\leq(\bar a, \beta)$; since $i \in \CA(x)$, we also have that $H_=(\bar a, \beta)$ supports $\epi(f)$ in $\left(x \atop {f(x)}\right)$. Thus, $a \in \partial f(x)$. Since $\partial f(x)$ is convex and closed by \cite[Theorem 23.4]{rockafellar}, the inclusion in \eqref{eq:subgradientInclusion} follows.
\end{proof}

\paragraph{} We are now ready to prove the main theorem of this section. The conditions of Theorem~\ref{theo:optCondHausdorffHomo} are also illustrated in Figure~\ref{fig:optCondHausdorff}.

\begin{figure}[htb]
\centering
\includegraphics[width=0.5\textwidth]{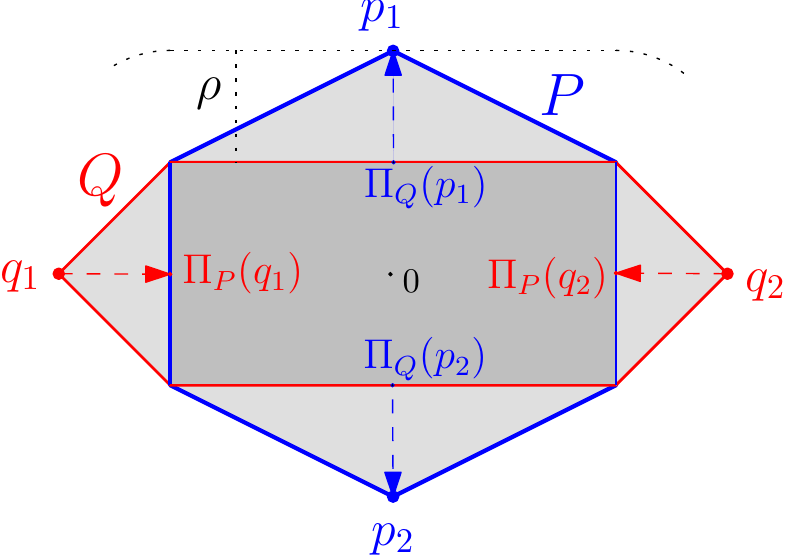}
\caption[Optimality condition for Hausdorff matching under homothetics]{The conditions of Theorem~\ref{theo:optCondHausdorffHomo}. The blue polytope $P\subseteq \R^2$ is in optimal position with respect to the red polytope $Q\subseteq \R^2$. Conditions (1) -- (3) of Theorem~\ref{theo:optCondHausdorffHomo} are verified with $R= \{p_1,p_2\}$, $S= \{q_1,q_2\}$ and $\rho$ as indicated. \label{fig:optCondHausdorff}}
\end{figure}

\begin{theo}[Optimality criterion for Hausdorff matching] \label{theo:optCondHausdorffHomo}%
Let $P, Q \in \CC_0^d$. Then, $P$ is in optimal homothetic position with respect to $Q$, if and only if there are $\rho\geq 0$ and $R\subseteq P, S\subseteq Q$ with $|R|+|S| \leq d+2$ such that the following three conditions hold:
\begin{enumerate}[(1)]
\item $ P\subseteq Q + \rho \B_2$ and $Q \subseteq P +\rho \B_2$ \label{item:primal}
%\item $ \|p - \Pi_Q(p)\|_2 = \rho ~\forall p\in R$ and $\|q - \Pi_P(q)\|_2=  \rho ~\forall q\in S$ \label{item:comp}
\item $ d_2(p,Q) = \rho ~\forall p\in R$ and $d_2(q,P)=  \rho ~\forall q\in S$ \label{item:comp}
\item $\displaystyle 0 \in \conv\left(\left\{\left({(p - \Pi_Q(p))^T p \atop{p - \Pi_Q(p)}}\right): p\in R\right\}\cup\left \{\left({(\Pi_P(q)- q)^T\Pi_P(q)} \atop {\Pi_P(q)- q}\right): q \in S \right\}\right)$. \label{item:dual}
\end{enumerate}
\end{theo}
\begin{proof}
If $P=Q$, then conditions (1) -- (3) are trivially necessary and sufficient with the choice $\rho=0$ and any $R,S \subseteq P$ with $|R|+|S|\leq d+2$.
We will henceforth assume that $P \neq Q$ and consequently $\rho > 0$. 

We first show the sufficiency of the conditions:
Assume that conditions (1) -- (3) hold for some $\rho >0$, $R \subseteq P$, and $S\subseteq Q$ with $|R|+|S|\leq d+2$. Let $f: (0,\infty)\times \R^d \rightarrow [0, \infty);  (\alpha, c)\mapsto \delta_2(\alpha P + c ,Q),$
and, as in Lemma~\ref{lem:fpfq}, let $ f_p(\alpha, c)= d_2(\alpha p + c, Q)$ for  $p\in P$, and 
$ f_q(\alpha, c)=d_2(q, \alpha P +c)$ for $q \in Q$.
Then, (1) and (2) imply
$$ \rho = f(1,0) = \max \bigl(\{ f_p(1,0): p \in P\} \cup \{f_q(1,0): q \in Q\}\bigr) $$

with $(R\cup S)\subseteq \CA\bigl((1,0)\bigr)$ in the notation of Lemma~\ref{lem:subgradientInclusion}. Hence, by Lemmas~\ref{lem:fpfq} and \ref{lem:subgradientInclusion}, condition (3) yields $0 \in \partial f(1,0)$ which in turn implies that $\delta_2(P,Q) \leq \delta_2(\alpha P + c, Q)$ for all $\alpha >0$ and $c \in \R^d$.

It remains to show that the conditions are also necessary. For this purpose, let $P$ be in optimal homothetic position with respect to $Q$. Choose $\rho:= \delta(P,Q) > 0$ and define 
$$R':=\{p\in P: d(p, Q) = \rho\}\subseteq \bd(Q+\rho\B_2) ~\text{and}~ S':= \{q\in Q: d(q, P)= \rho\}\subseteq \bd(P+\rho \B_2).$$ 
We have $R' \neq \emptyset$ and $S' \neq \emptyset$, because one of these sets being empty would imply that $P$ is not optimally scaled. 
For $x \in \R^d$ and $C \in \CC_0^d$ define further $ a_C(x):= x- \Pi_C(x)$ % \quad \text{ and }  \quad a_P(x):= x-\Pi_P(x) $$ 
and let
$$A:= \left\{\left({a_Q(p)^Tp} \atop{a_Q(p)}\right) : p \in R'\right\}, \quad B:=\left\{\left({a_P(q)^T \Pi_P(q)} \atop a_P(q)\right): q \in S'\right\}.$$ We show $0\in \conv (A\cup (-B))$. Carathéodory's Theorem (e.g.~\cite{dgk-63}) then yields that $R'$ and $S'$ can be reduced to subsets $R\subseteq R'$, $S\subseteq S'$  with $|R|+|S| \leq d+2$ .

For a contradiction, assume that $0 \not \in \conv (A\cup (-B))$. 
Then, $0$ can be strictly separated from $\conv(A\cup (-B))$, i.e.\ there exists $(\eta, y^T)^T \in\R^{d+1}\setminus\{0\}$ such that 
\begin{equation}
 \bigl(a_Q(p)^T p\bigr)\eta + a_Q(p)^T y \leq -1~\forall p \in R' \quad \text{ and }\quad   \bigl(a_P(q)^T\Pi_P(q)\bigr)\eta + a_P(q)^T y \geq 1~\forall q\in S'.
\label{eq:separation}
\end{equation} 
We will show that there exists $\lambda> 0$ such that $ \delta\bigl((1+\lambda \eta) P +\lambda y , Q\bigr) < \delta(P,Q)$, which contradicts the optimality of $(1,0)$.

For this purpose, let $$Q_\rho:= Q + \rho \B_2 = \bigcap_{x\in \bd(Q_\rho)} H_{\leq} (a_Q(x), a_Q(x)^T x),$$ 
where the second presentation is obtained by Lemma~\ref{lem:outerParallel}. 

Define $N:= \left\{(x, p) \in \bd(Q_\rho)\times P:   \bigl(a_Q(x)^Tp\bigr)\eta   + a_Q(x)^T y \geq 0\right\}$
and $f: N \rightarrow \R; (x,p) \mapsto a_Q(x)^T x - a_Q(x)^Tp$. For $(x,p) \in N$, \eqref{eq:separation} yields $\left(a_Q(x)^T p ,a_Q(x)^T\right)^T \notin A$, which implies $p \neq x$. Together with $p \in Q+ \rho \B_2$ , this yields $a_Q(x)^T p < a_Q(x)^T x$. Thus, for all $(x,p) \in N$, $f(x,p) >0$. Since $f$ is continuous and $N$ is compact, there exists $\varepsilon_1 >0$ such that 
\begin{equation}
f(x,p) \geq \varepsilon_1 ~\forall (x,p) \in N
\label{eq:fN>0}
\end{equation}

Now, for $(x,p) \in \bigl(\bd(Q_\rho)\times P\bigr) \setminus N$ and $\lambda >0$, we have
$$a_Q(x)^T\bigl((1+\lambda \eta) p + \lambda y\bigr) = \underbrace{a_Q(x)^T p}_{\leq a_Q(x)^T x} + \lambda \underbrace{\bigl((a_Q(x)^T p)\eta + a_Q(x)^T y\bigr)}_{<0} <a_Q(x)^T x .$$
By \eqref{eq:fN>0} and the boundedness of $Q_\rho$ and $P$, we can choose $\lambda >0 $ sufficiently small such that for all $(x,p) \in N$ 
$$a_Q(x)^T((1+\lambda \eta) p + \lambda y) = \underbrace{a_Q(x)^T p}_{\leq a_Q(x)^T x -\varepsilon_1} + \lambda \bigl((a_Q(x)^T p)\eta + a_Q(x)^T y\bigr) <a_Q(x)^T x .$$

Thus, 
\begin{equation}
(1+\lambda\eta)P + \lambda y \subseteq \int(Q_\rho) \quad \text{for all } \lambda >0 \text{ sufficiently small.}
\label{eq:PinQrho}
\end{equation} 

Assume further that $\lambda > 0$ is sufficiently small such that $\lambda \eta > -1$ and let
$$  P + \frac{\rho}{1+\lambda\eta} \B_2 = \bigcap_{x \in \bd (P+\rho \B_2)} H_\leq \left(a_P(x),a_P(x)^T \Pi_P(x) + \frac{\rho^2}{1+\lambda\eta}\right),$$
where again the $\CH$-presentation is obtained via Lemma~\ref{lem:outerParallel}.

Define $M:= \left\{x \in \bd(P+\rho\B_2):  \bigl(a_P(x)^T\Pi_P(x)\bigr)\eta   + a_P(x)^T y \leq 0\right\}$
and $g: M\times Q \rightarrow \R; \quad (x,q) \mapsto a_P(x)^T \Pi_P(x) +\rho^2 - a_P(x)^T q.$ For $x \in M $, the separation property \eqref{eq:separation} yields $\left(a_P(x)^T \Pi_P(x) ,a_P(x)^T\right)^T \notin B $, which implies that there is no $q\in Q$ such that $a_P(x)^T q = a_P(x)^T\Pi_P(x) +\rho^2$. Together with $Q \subseteq P +\rho \B_2$, we obtain $g(x,q) > 0$ for all $x \in M$ and $q \in Q$. Since $M\times Q$ is compact and $g$ continuous, there exists $\varepsilon_2 >0$ such that $g(x,q) \geq \varepsilon_2$ for all $x \in M$ and $q \in Q$.

Hence, we can again choose $\lambda >0 $ sufficiently small such that for all $x\in M$  and $q \in Q$, we have 
\begin{equation} 
a_P(x)^T q < a_P(x)^T \Pi_P(x) + \rho^2 + \lambda\bigl( (a_P(x)^T \Pi_P(x))\eta + a_P(x)^T y\bigr)
\label{eq:QinPlambdaRhoAlg}
\end{equation}
and \eqref{eq:QinPlambdaRhoAlg} is also fulfilled for $x \in \bd(P +\rho \B_2)\setminus M$, $\lambda >0$ and $q\in Q$.
Rearranging \eqref{eq:QinPlambdaRhoAlg} shows that it is equivalent to 
\begin{equation}
\frac{1}{1+\lambda\eta} (Q -\lambda y) \subseteq \int \left(P +\frac{\rho}{1+\lambda\eta} \B_2\right) \Longleftrightarrow Q \subseteq \int\left((1+\lambda\eta)P + \lambda y +\rho \B_2\right)
\label{eq:QinPlambdaRho}
\end{equation}
Together, \eqref{eq:PinQrho} and \eqref{eq:QinPlambdaRho} yield the desired contradiction.

\end{proof}

\subsection{Helly-Type Properties}
\label{sec:hausdorffHelly}

\paragraph{} Theorem~\ref{theo:optCondHausdorffHomo} at hand, we can derive Helly-properties of the Hausdorff matching problem. For the case of $\CV$-polytopes, this question has already received attention and it was shown as a byproduct in \cite{amentaHausdorff} that Hausdorff matching with $\CV$-polytopes can be formulated as a Convex Program, implying the standard Helly-type theorem on the constraints of the Convex Program. %in order to show that it is a Generalized Linear Program \cite{GLP2, GLP} of combinatorial dimension $\delta= d+2$. 
With Theorem~\ref{theo:optCondHausdorffHomo} at hand, we can generalize this result to arbitrary convex bodies.

\begin{cor}[Helly-type theorem for Problem~\ref{prb:HMHomo}] \label{cor:HellyTypeHausdorffMatching}%
Let $P,Q \in \CC_0^d$ and for $R\subseteq P$, $S\subseteq Q$, define 
\begin{equation}
\begin{array}{lllcll}
\rho(R,S):=& \min 	& \rho \\
&s.t.	& \alpha p +c & \in &Q + \rho \B_2 	&\forall p \in R \\
		 & & q& \in &\alpha P + c + \rho \B_2	&\forall q \in S\medskip\\
%		&	& c  \in \R^d \\
%		& & \alpha, \rho \geq  0.
\end{array}
\label{eq:HausdorffMatchingVVHelly}
\end{equation}
Then for any $\rho^* \geq 0$,
$$\rho(P, Q)\leq \rho^* \quad \Longleftrightarrow \quad  \rho(R,S) \leq \rho^* ~~\forall R\subseteq P, S\subseteq Q, |R|+|S| \leq d+2.$$
\end{cor}

\paragraph{} In addition to restricting the number of constraints, we can also state a Helly-type theorem which shows that for any two convex bodies there are always finite subsets of points the convex hulls of which induce the same optimal solution as the bodies themselves. In the same way as for the Minimum Enclosing Ball problem (cf.~\cite[Lemma 2.2]{coresetPaperDCG}), the size bound on these sets only depends on the ambient dimension and is independent of the bodies $P$ and $Q$.

\begin{theo}[0-core-sets for Problem~\ref{prb:HMHomo}]\label{theo:HMHomoHelly}%
Let $P,Q \in \CC_0^d$. There are subsets $R\subseteq \ext(P)$ and $S\subseteq \ext(Q)$ such that $|R|+|S|\leq d(d+2)$ and 
$$\delta_H(\conv(R),\conv(S))= \delta_H(P,Q).$$ 
\end{theo}
\begin{proof}
Assume without loss of generality that $\delta_H(P,Q)=\delta(P,Q)$.
Corollary~\ref{cor:HellyTypeHausdorffMatching} implies that there are subsets $R' \subseteq P$ and $S'\subseteq Q$ with $|R'|+|S'|\leq d+2$ such that in \eqref{eq:HausdorffMatchingVVHelly} only the containment constraints for $p \in R'$ and $q \in S'$ are necessary.

For $p\in R'$, let $a:= p-\Pi_Q(p)$ and $\beta:=a^T \Pi_Q(p)$. We have that $Q\subseteq H_\leq(a,\beta)$  and $\Pi_Q(p) \in H_=(a,\beta).$  Carathéodory's Theorem (e.g.~\cite{dgk-63}) applied to $Q \cap H_=(a,\beta)$ yields the existence of $q_1^p,\dots, q_d^p \in \ext \bigl(Q \cap H_=(a,\beta)\bigl) \subseteq \ext(Q)$ such that $\Pi_Q(p) \in \conv\{q_1^p,\dots, q_d^p\}$. Defining 
$$ S:= S' \cup \bigcup_{p\in R'} \{q_1^p,\dots, q_d^p\} $$ 
assures that $d(p, \conv(S))= d(p, Q)$ for all $p \in R'$.
Applying the same argument to $R'$ shows that $|R|+|S| \leq d(d+2)$.
\end{proof}

\paragraph{Remark.} The restriction to extreme points in Theorem~\ref{theo:HMHomoHelly} aims at the algorithmic application, where $P$ and $Q$ are specified as $\CV$-polytopes and $\ext(P)$ and $\ext(Q)$ are easily accessible. If one drops this restriction, the bound $|S|+|R|\leq d(d+2)$ can be improved to $|S|+|R|\leq 2(d+2)$.

\subsection{Exact Algorithms and Approximations}
\label{sec:hausdorffAlgos}
\paragraph{} The following theorem shows that, if the input polytopes are specified in $\CV$-presentation, also the matching problem can be solved efficiently.

\begin{theo} \label{theo:matchingSOCP}%
Let $P:= \conv\{p_1,\dots, p_n\}\subseteq \R^d$ and $Q:= \conv\{q_1,\dots, q_m\}\subseteq \R^d$.  If $\B= \B_2$, then $\delta_H(P,Q)$ can be approximated to any accuracy. If a $\CV$- or $\CH$-presentation of $\B$ is available, $\delta_H(P,Q)$ can be computed exactly in polynomial time. 
\end{theo}
\begin{proof}
The optimal solution of the Second Order Cone Program in \eqref{eq:matchingSOCP} can be approximated to any accuracy and equals $\delta_H(P,Q)$:

\begin{equation}
\begin{array}{llllll}
\delta_H(P,Q)= 	&\min	& \rho \\
				&s.t.	&  \Bigl\| \alpha p_i + c - \sum\limits_{j=1}^m \lambda_{ij} q_j \Bigr \|_2 	&	\leq & \rho	& \forall i \in [n] \\
				&		&  \Bigl\| q_j - \sum\limits_{i=1}^n \mu_{ij} p_i - c \Bigr\|_2 	&	\leq & \rho	& \forall j \in [m] \\ 
				&		&   \sum\limits_{j=1}^m \lambda_{ij}  =  1& & & \forall i\in [n] \\
				&		&   \sum\limits_{i=1}^n \mu_{ji}  =  \alpha & && \forall j\in [m] \\
				&		&   \lambda_{ij}, \mu_{ji} \geq 0	&&& \forall i \in [n], j\in [m].
\end{array}
\label{eq:matchingSOCP}
\end{equation}
If $\B$ is a polytope in $\CV$- or $\CH$-presentation, one may replace the Second Order Cone constraints in \eqref{eq:matchingSOCP} by linear constraints  as in \eqref{eq:distVH} and \eqref{eq:distVV} and obtain a Linear Program for the computation of $\delta_H(P,Q)$.
\end{proof}

\paragraph{} By Theorem~\ref{theo:HausdorffHHard}, simple evaluation of the Hausdorff distance is W[1]-hard, if at least one $\CH$-presented polytope is involved. Nonetheless, we can apply the so-called \enquote{Reference Point Method} from \cite{hausdorffReferencePoints} and obtain a polynomial time approximation algorithm also for this case. In the original setting, this framework requires the computation of the diameter of the involved sets, which is also W[1]-hard for $\CH$-polytopes \cite{normmaxW1}. We therefore replace the diameter of a polytope by the diameter of its bounding box. The performance guarantee of this version is slightly worse but in the same order of magnitude as the original result in \cite{hausdorffReferencePoints}.

\begin{defi}[Reference points]
For $P \in \CC_0^d$, let \index{reference point}\index{$r(P)$}\index{$s(P)$}\index{$d(P)$}
$$r(P):=  (h(P, -e_1), \dots,  h(P, -e_d))\in \R^d, \quad s(P):= (h(P, e_1),\dots, h(P, e_d))^T  \in \R^d $$
the \enquote{lower left} and \enquote{upper right} corner of the bounding box of $P$ and 
$$D(P):= \|s(P)- r(P)\|_2$$
the diameter of this box.
\end{defi}

\paragraph{} We first state some elementary properties of the interplay between the reference points of $P$ and $Q$ and the Hausdorff distance $\delta_2(P,Q)$.

\begin{lem}[Properties of the reference points] \label{lem:propertiesReferencePoints}%
\vspace{-0.6cm}
\begin{enumerate}[a)]
\item Let $P \in \CC_0^d$ and $\alpha_1,\alpha_2 >0$.\label{item:prop1} Then, $\delta_2\bigl(\alpha_1(P- r(P)), \alpha_2(P-r(P))\bigr) \leq |\alpha_1-\alpha_2|D(P)$. 
\item Let $P,Q \in \CC_0^d$. Then $|D(P)- D(Q)| \leq 2 \sqrt{d}\delta_2(P, Q)$. \label{item:prop2}
\item Let $P,Q \in \CC_0^d$. Then $\|r(P) - r(Q)\|_2 \leq \sqrt{d} \delta_2(P,Q)$.\label{item:prop3}
\end{enumerate}
\end{lem}
\begin{proof}
\begin{enumerate}[a)]
\vspace{-0.8cm}
\item Let without loss of generality $r(P)=0$, and $p\in P$. Then, $d(\alpha_1 p,\alpha_2 P) \leq \|\alpha_1 p -\alpha_2 p\|_2 \leq |\alpha_1-\alpha_2|D(P)$ and the same argument shows that also $\max \{d(x, \alpha_1 P): x \in \alpha_2 P\} \leq |\alpha_1-\alpha_2|D(P)$.
\item By Lemma~\ref{lem:hausdorff3}, for $i \in [d]$, we have $|h(P, \pm e_i) - h(Q, \pm e_i)| \leq \delta_2(P,Q)$ and hence $|D(P)-D(Q)|\leq 2\sqrt{d}\delta_2(P,Q)$.
\item The statement also follows from $|h(P, - e_i) - h(Q, - e_i)| \leq \delta_2(P,Q)$ for all $i\in [d]$. 
\end{enumerate}
\end{proof}

\paragraph{} If we compute a homothetic transformation which only relies on our two reference points, we can give the following performance guarantee.

\begin{lem}[Approximation by reference points] \label{lem:referencePoints}%
Let $P,Q \in \CC_0^d$ and 
$$\bar \alpha:= \frac{D(Q)}{D(P)}  \quad \text{ and } \quad \bar c:=r(Q) - \bar \alpha r(P)$$
Then, 
$$ \delta_2\bigl(\bar \alpha P + \bar c,~ Q\bigr) \leq\left( 3 \sqrt{d}+1\right) \delta_H(P,Q). $$
\end{lem}
\begin{proof}
Let $\alpha^* > 0 $, $c^* \in \R^d$ and $\rho^* \geq 0$ such that $\rho^*= \delta_H(P,Q)= \delta_2(\alpha^* P + c^*, Q)$. Then, by Lemma~\ref{lem:propertiesReferencePoints}\ref{item:prop3}), 
\begin{equation}
%\delta_2(\alpha^* (P - r(P)) + r(Q), Q)= 
\delta_2\Bigl(\alpha^* P + c^* + r(Q) - r(\alpha^* P +c^*),~ Q\Bigr)\leq \left(1 + \sqrt{d}\right)\rho^*.
\label{eq:SSchlangeSopt}
\end{equation}
Moreover, by using Lemma~\ref{lem:propertiesReferencePoints}\ref{item:prop1}) and \ref{item:prop2}), 
\begin{equation}
{\renewcommand{\arraystretch}{1.5}
\begin{array}{rl} 
& \delta_2\Bigl(\bar \alpha(P - r(P)) + r(Q), \alpha^* (P - r(P)) + r(Q)\Bigr) \\
\leq  & |\alpha^* - \bar\alpha|D(P) = |\alpha^* D(P+c^*) - D 	(Q)| \leq 2\sqrt{d} \rho^*.
\label{eq:SSschlange}
\end{array}}
\end{equation}

By combining \eqref{eq:SSchlangeSopt} and \eqref{eq:SSschlange}, we obtain
\begin{equation}
\begin{array}{rcl}
 \delta_2( \bar \alpha P + \bar c, Q) &\leq & ~~~ \delta_2\Bigl(\bar \alpha(P - r(P)) + r(Q), ~\alpha^* (P - r(P)) + r(Q)\Bigr)  \\
& & + ~\delta_2\Bigl(\alpha^* P + c^* + r(Q) - r(\alpha^* P +c^*), ~Q\Bigr) \\
&  \leq& {(3\sqrt{d}+1)\rho^*}.  
\end{array}
\nonumber
\end{equation}
\end{proof}

\paragraph{} Since an axis-parallel bounding box for a polytope in $\CV$- or $\CH$-presentation can be computed in polynomial time \cite{gk-93}, we obtain the following.

\begin{theo}\label{theo:referencePoints}%
For polytopes $P,Q \subseteq \R^d$ in $\CV$- or $\CH$-presentation, a factor-$\left( 3 \sqrt{d}+1\right)$-approximation of $\delta_H(P,Q)$ can be computed in polynomial time.
\end{theo}

\paragraph{Acknowledgements.} The author would like to thank René Brandenberg, Peter Gritzmann, Andreas Schulz, and Günter Rote for inspiring discussions and many helpful comments.

\bibliographystyle{plain}
\bibliography{references}
\end{document}